\documentclass[11pt]{article}
\usepackage{times}

\usepackage[margin=1in]{geometry}
\usepackage{graphicx}
\usepackage[pdftex,colorlinks=true,linkcolor=blue,citecolor=blue,urlcolor=black]{hyperref}
\usepackage{amsmath, amsthm, amssymb}
\usepackage{subfigure}
\usepackage{comment}
\usepackage{url}
\usepackage[ruled,lined,linesnumbered]{algorithm2e}
\usepackage{tikz}

\newcommand{\ket}[1]{| #1 \rangle}
\newcommand{\bra}[1]{\langle #1|}
\newcommand{\braket}[1]{\langle #1 \rangle}

\newcommand{\proj}[1]{| #1 \rangle \langle #1 |}

\DeclareMathOperator{\polylog}{polylog}

\DeclareMathOperator{\tr}{tr}

\newcommand{\be}{\begin{equation}}
\newcommand{\ee}{\end{equation}}
\newcommand{\bea}{\begin{eqnarray}}
\newcommand{\eea}{\end{eqnarray}}
\newcommand{\bes}{\begin{equation*}}
\newcommand{\ees}{\end{equation*}}
\newcommand{\beas}{\begin{eqnarray*}}
\newcommand{\eeas}{\end{eqnarray*}}

\makeatletter
\newtheorem*{rep@theorem}{\rep@title}
\newcommand{\newreptheorem}[2]{%
\newenvironment{rep#1}[1]{%
 \def\rep@title{#2 \ref{##1} (restated)}%
 \begin{rep@theorem}}%
 {\end{rep@theorem}}}
\makeatother

\newtheorem{thm}{Theorem}
\newtheorem*{thm*}{Theorem}

\newtheorem{lem}[thm]{Lemma}
\newtheorem*{lem*}{Lemma}

\newreptheorem{thm}{Theorem}
\newreptheorem{lem}{Lemma}
\newreptheorem{cor}{Corollary}

\newcommand{\FIXME}[1]{}

\begin{document}


\title{Pretty simple bounds on quantum state discrimination}
\author{Ashley Montanaro\thanks{School of Mathematics, University of Bristol, UK; {\tt ashley.montanaro@bristol.ac.uk}.}}
\maketitle

\begin{abstract}
We show that the quantum measurement known as the pretty good measurement can be used to identify an unknown quantum state picked from any set of $n$ mixed states that have pairwise fidelities upper-bounded by a constant below 1, given $O(\log n)$ copies of the unknown state, with high success probability in the worst case. If the unknown state is promised to be pure, there is an explicit measurement strategy which solves this worst-case quantum state discrimination problem with $\widetilde{O}(\|G\|)$ copies, where $G$ is the Gram matrix of the states.
\end{abstract}


\section{Introduction}

A fundamental task in quantum information theory is {\em quantum state discrimination}. Here we will be concerned with the following variant of this problem: We are given an unknown state $\rho$ picked from a known set $S = \{ \rho_i \}$ of mixed states, where $|S|=n$, and our task is to identify $\rho$ with the lowest possible worst-case probability $\delta$ of failure. That is, we want to find a quantum measurement (POVM), described by a set of positive semidefinite operators $\mu_i$ with $\sum_i \mu_i = I$, such that $\max_i 1 - \tr \mu_i \rho_i$ is minimised. This task has been termed ``minimax'' quantum state discrimination~\cite{dariano05}, though here we will refer to it as {\em worst-case} quantum state discrimination. We will also consider the closely related question where we are given $\delta > 0$ in advance, and would like to determine the number $k$ of copies of $\rho$ that are required to achieve failure probability $\delta$ by performing a measurement on $\rho^{\otimes k}$.

Quantum state discrimination has been a topic of intensive study within quantum information theory (see~\cite{barnett09,bergou10,bae15} for reviews), although the majority of works consider the setting where each state $\rho_i$ is produced with a known probability $p_i$. However, in the context of a quantum algorithm which should have a low worst-case probability of failure, worst-case discrimination is often the most natural setting.

In the worst-case setting, it was shown by Harrow and Winter~\cite{harrow12a} that, if all states in $S$ have pairwise fidelities $F(\rho_i,\rho_j) := \| \sqrt{\rho_i} \sqrt{\rho_j}\|_1$ upper-bounded by $F$ (where $0 < F < 1$), then the worst-case state discrimination problem can be solved with $O(\log(n/\delta)/\log(1/F))$ copies of $\rho$. This result is nonconstructive, and the proof proceeds via a minimax theorem. That is, existence of a measurement that solves the worst-case state discrimination problem is shown, without describing the measurement explicitly.

Here we show that there is an explicit measurement, the pretty good measurement~\cite{belavkin75,holevo79,hausladen94} (``PGM'', defined below) which achieves a similar scaling of the number of copies:

\begin{thm}
\label{thm:mixed}
If $F(\rho_i, \rho_j) \le 1-\epsilon$ for all pairs of distinct states $\rho_i,\rho_j \in S$, then the worst-case state discrimination problem can be solved with failure probability $\delta$ by applying the PGM to $O(\log (n/\delta)/\epsilon)$ copies of $\rho$.
\end{thm}

Next we show that, in the case of pure states, this result can sometimes be improved to a scaling of the number of copies required which does not depend on $n$. In the case where all states in $S$ are pure, write $\rho_i = \proj{\psi_i}$, and let $G$ be the Gram matrix of these vectors, i.e.\ $G_{ij} = \braket{\psi_i|\psi_j}$. Then:

\begin{thm}
\label{thm:pure}
If all states in $S$ are pure, then the worst-case state discrimination problem can be solved with success probability at least $\|G\|^{-1}$ by applying the PGM to one copy of $\rho$. If additionally $\tr \rho_i \rho_j \le 1-\epsilon$ for all pairs of distinct states in $S$, then there is an explicit measurement strategy which, applied to
\[ O((\|G\|/ \epsilon) (\log 1/\delta) \log(\|G\|/\delta)) = \widetilde{O}(\|G\|/\epsilon) \]
 copies of $\rho$, solves the worst-case state discrimination problem with failure probability $\delta$.
\end{thm}

In Theorem \ref{thm:pure} and throughout, we use $\widetilde{O}(f(\|G\|,\epsilon,\delta))$ to denote $O(f(\|G\|, \epsilon, \delta) \polylog(\|G\|,1/\epsilon,1/\delta))$.

Thus the operator norm of $G$ bounds the success probability and the number of copies required to solve the worst-case quantum state discrimination problem. To gain some intuition for this result, note that if all the states in $S$ are orthogonal, $G$ is the identity matrix, so $\|G\| = 1$; whereas if all the states in $S$ are equal, $G_{ij} = 1$ and $\|G\| = n$.

Theorem \ref{thm:pure} can be applied, for example, to random pure states picked from a variety of distributions. It often holds that, for random states in $d$ dimensions and with $n=O(d)$, $\|G\| = O(1)$ with high probability. Indeed, this holds for any states whose amplitudes with respect to an arbitrary basis are close to iid random variables with suitably bounded 4th moments~\cite{yin88,bai99}. Examples are Haar-random pure states and states of the form $\frac{1}{\sqrt{d}} \sum_{i=1}^d z_i \ket{i}$, where $z_i$ is uniformly randomly chosen from $\{\pm 1\}$. Upper bounds were proven on the probability of failure of discriminating these ensembles of states in~\cite{montanaro07a} in the case where there is a uniform probability distribution on the states in $S$. Theorem \ref{thm:pure} extends this to a worst-case setting.


\section{Definitions}

The pretty good measurement~\cite{belavkin75,holevo79,hausladen94} (PGM), also known as the square-root measurement~\cite{hausladen96}, is defined as follows: for each state $\rho_i$ we introduce a measurement operator $\mu_i = \Sigma^{-1/2} \rho_i \Sigma^{-1/2}$, where $\Sigma := \sum_i \rho_i$ and the inverse is taken on the support of $\Sigma$. This is a valid POVM as
\[ \sum_i \mu_i = \sum_i \Sigma^{-1/2} \rho_i \Sigma^{-1/2} = \Sigma^{-1/2} \left( \sum_i \rho_i \right) \Sigma^{-1/2} = I, \]
where the identity operator is with respect to the span of the states in the support of $S$. Note that often this measurement is defined in terms of states $\rho_i$ normalised by some a priori probabilities $p_i$, but here these are not used.

Write $\rho_i = \sum_j \lambda_{ij} \proj{\psi_{ij}}$ for the eigendecomposition of $\rho$, and let $G$ be the Gram matrix of the weighted states $\{\sqrt{\lambda_{ij}} \ket{\psi_{ij}}\}$. $G$ has a natural block structure in terms of $i$. If we define the vectors $\ket{\mu_{ij}} = \Sigma^{-1/2} \sqrt{\lambda_{ij}} \ket{\psi_{ij}}$ and the positive semidefinite matrix $P_{ik,jl} = \sqrt{\lambda_{jl}} \braket{\mu_{ik}|\psi_{jl}}$, then
\beas
(P^2)_{ik,jl} &=& \sum_{r,s} \sqrt{\lambda_{ik}} \lambda_{rs} \sqrt{\lambda_{jl}} \braket{\psi_{ik}|\Sigma^{-1/2}|\psi_{rs}}\braket{\psi_{rs}|\Sigma^{-1/2}|\psi_{jl}}\\
&=& \sqrt{\lambda_{ik}} \sqrt{\lambda_{jl}} \bra{\psi_{ik}} \Sigma^{-1/2}\left( \sum_{r,s} \lambda_{rs} \ket{\psi_{rs}} \bra{\psi_{rs}} \right) \Sigma^{-1/2}\ket{\psi_{jl}}\\
&=& G_{ik,jl}.
\eeas
Thus the probability that the PGM outputs $i$ on input $\rho_j$ is
\[ \tr \mu_i \rho_j = \tr \left( \sum_k \proj{\mu_{ik}}\right) \left( \sum_l \lambda_{jl} \proj{\psi_{jl}}\right)  = \sum_{k,l} \lambda_{jl} |\braket{\mu_{ik}|\psi_{jl}}|^2 = \|P^{(ij)}\|_2^2 = \|\sqrt{G}^{(ij)}\|_2^2, \]
where we use $P^{(ij)}$ to denote the $(i,j)$'th block of $P$ and $\|M\|_2^2 := \sum_{i,j} |M_{ij}|^2$. Write
\[ P_E(S) := \max_i 1 - \tr \mu_i \rho_i = \max_i 1 - \|\sqrt{G}^{(ii)}\|_2^2\]
for the worst-case probability of error when the PGM is used.


\section{Worst-case bounds for mixed states}

It was shown in~\cite{harrow12a}, based on a bound of~\cite{barnum02}, that $O(\log n)$ copies are required to identify a state picked from an arbitrary set of $n$ states with pairwise fidelities bounded above by a constant $F<1$. The result of~\cite{harrow12a} states that there exists a measurement that achieves this complexity, without describing that measurement explicitly. The reason is that the result of~\cite{barnum02} is stated in terms of a known probability distribution on the states, and~\cite{harrow12a} uses a minimax theorem to lift this to a worst-case bound. Here we show that the PGM itself achieves such a bound. The proof is directly analogous to that of~\cite{barnum02} in terms of considering Hilbert-Schmidt norms of off-diagonal blocks of $\sqrt{G}$, though it proceeds via a slightly different route.


\begin{lem}
\label{lem:pe}
\[ P_E(S) \le \sum_{i \neq j} F(\rho_i,\rho_j). \]
\end{lem}

\begin{proof}
Set $\Lambda = \sum_i \proj{i} \otimes G^{(ii)}$, $\Delta = \sum_{i \neq j} \ket{i} \bra{j} \otimes G^{(ij)}$, such that $G = \Lambda + \Delta$. We are interested in upper-bounding
\[ P_E(S) = \max_i \sum_{j\neq i} \|\sqrt{G}^{(ij)}\|_2^2 \le \sum_{i,j:i\neq j} \|\sqrt{G}^{(ij)}\|_2^2 \le \| \sqrt{G} - \sqrt{\Lambda}\|_2^2, \]
where the second inequality holds because $\sqrt{\Lambda}^{(ij)}=0$ for $i \neq j$.
For any unitarily invariant norm $\|\cdot\|$ and any positive operators $A$, $B$, we have~\cite[Theorem X.1.3]{bhatia97}
\[ \|\sqrt{A} - \sqrt{B} \| \le \| \sqrt{|A-B|} \|. \]
Hence
\[ \|\sqrt{G} - \sqrt{\Lambda} \|_2^2 = \|\sqrt{\Lambda + \Delta} - \sqrt{\Lambda} \|_2^2 \le \| \sqrt{|\Delta|} \|_2^2 = \|\Delta\|_1 \le \sum_{i \neq j} \|G^{(ij)}\|_1 = \sum_{i \neq j} F(\rho_i,\rho_j), \]
where $\|\cdot\|_1$ denotes the trace norm. The last equality follows from
\[ G^{(ij)} = \sum_{k,l} \sqrt{\lambda_{ik}} \sqrt{\lambda_{jl}} \braket{\psi_{ik}|\psi_{jl}} \ket{k}\bra{l} = \left(\sum_k \sqrt{\lambda_{ik}} \ket{k} \bra{\psi_{ik}} \right) \left( \sum_l \sqrt{\lambda_{jl}} \ket{\psi_{jl}} \bra{l}  \right), \]
which implies that
\[ \|G^{(ij)}\|_1 = \left\| \left(\sum_k \sqrt{\lambda_{ik}} \ket{\psi_{ik}} \bra{\psi_{ik}} \right) \left( \sum_l \sqrt{\lambda_{jl}} \ket{\psi_{jl}} \bra{\psi_{jl}}  \right) \right\|_1 = \|\sqrt{\rho_i}\sqrt{\rho_j}\|_1 = F(\rho_i,\rho_j) \]
by unitary invariance of the trace norm and orthonormality of the states $\{\ket{\psi_{ik}}\}$ for each $i$.
\end{proof}

Lemma \ref{lem:pe} implies that if $F(\rho_i,\rho_j) \le 1/(3n^2)$ for all $i\neq j$, then $P_E(S) \le 1/3$. This can be seen as a generalisation of a folklore result proven by Ambainis and de Wolf~\cite{ambainis14c}, albeit with a somewhat worse constant. The result of~\cite{ambainis14c} was only shown for pure states, but states that if $F(\rho_i,\rho_j) \le 1/n^2$ for all $i\neq j$, then $P_E(S) \le 1/3$.

\begin{repthm}{thm:mixed}
If $F(\rho_i, \rho_j) \le 1-\epsilon$ for all pairs of distinct states $\rho_i,\rho_j \in S$, then the worst-case state discrimination problem can be solved with failure probability $\delta$ by applying the PGM to $O(\log (n/\delta)/\epsilon)$ copies of $\rho$.
\end{repthm}

\begin{proof}
Let $S' = \{ \rho_i^{\otimes k} : i \in \{1,\dots,n\} \}$; then by Lemma \ref{lem:pe},
\[ P_E(S') \le \sum_{i \neq j} F(\rho_i^{\otimes k},\rho_j^{\otimes k}) \le n(n-1) \max_{i\neq j} F(\rho_i, \rho_j)^k \le n^2 (1-\epsilon)^k \le n^2 e^{-k\epsilon} \]
so it is sufficient to take $k = \lceil (2/\epsilon) \ln (n/\delta) \rceil$ to achieve failure probability at most $\delta$.
\end{proof}


\section{Improved bounds for pure states}

We now find an alternative bound which is only good for pure (or not too mixed) states, but which further improves on~\cite{harrow12a} by not having any (explicit) dependence on $n$. The bound states that, for any set of pure states whose pairwise fidelities are bounded above by $1-\epsilon$, where $\epsilon > 0$ is a constant, the state discrimination problem can be solved with $\widetilde{O}(\|G\|)$ copies of the unknown state, where $\|\cdot\|$ is the operator norm. The bound is based on the following technical lemma:

\begin{lem}
\label{lem:bound1}
For all $i$,
\[ \|G\|^{-1} \tr \rho_i^2 \le \tr \mu_i \rho_i \le \|G^{-1}\| \tr \rho_i^2. \]
\end{lem}

\begin{proof}
We have
\[ \lambda_{\min}(\sqrt{G}) \sqrt{G} \le G \le \lambda_{\max}(\sqrt{G}) \sqrt{G} \]
in a positive semidefinite sense, where $\lambda_{\min}(G)$, $\lambda_{\max}(G)$ are the minimal and maximal eigenvalues of $G$. This inequality is preserved under projections and taking the 2-norm. So
\[ \|G\| \tr \mu_i \rho_i = \lambda_{\max}(\sqrt{G})^2 \|\sqrt{G}^{(ii)}\|_2^2 \ge \|G^{(ii)}\|_2^2 = \tr \rho_i^2, \]
which is the lower bound of the lemma, and
\[ \|G^{-1}\|^{-1} \tr \mu_i \rho_i = \lambda_{\min}(\sqrt{G})^2 \|\sqrt{G}^{(ii)}\|_2^2 \le \|G^{(ii)}\|_2^2 = \tr \rho_i^2, \]
which is the upper bound.
\end{proof}

If one assumes a uniform distribution on the states in $S$ and that they are pure, the lower bound in Lemma \ref{lem:bound1} is a corollary of \cite[Lemma 2]{montanaro07a}.

\begin{repthm}{thm:pure}
If all states in $S$ are pure, then the worst-case state discrimination problem can be solved with success probability at least $\|G\|^{-1}$ by applying the PGM to one copy of $\rho$. If additionally $\tr \rho_i \rho_j \le 1-\epsilon$ for all pairs of distinct states in $S$, then there is an explicit measurement strategy which, applied to
\[ O((\|G\|/ \epsilon) (\log 1/\delta) \log(\|G\|/\delta)) = \widetilde{O}(\|G\|/\epsilon) \]
 copies of $\rho$, solves the worst-case state discrimination problem with failure probability $\delta$.
\end{repthm}

\begin{proof}
The first part is immediate from Lemma \ref{lem:bound1}. For the second part, apply the PGM separately to $k = \lceil \|G\| \ln (2/\delta) \rceil$ copies of $\rho$, obtaining outcomes $i_1,\dots,i_k$. The probability that the outcome corresponding to $\rho$ is not among the outcomes obtained is at most
\[ (1- \|G\|^{-1})^k \le e^{-\|G\|^{-1} k } \le \delta/2. \]
Then copies of $\rho$ are tested via $l = \lceil \ln (2k/\delta) / \epsilon \rceil$ uses of each of the accept/reject measurements that project onto $\rho_{i_j}$, for each outcome $i_j$. If all measurements in the $j$'th group accept, then the protocol outputs $i_j$. If no such group of measurements all accept, the protocol outputs ``fail''.

Each accept/reject measurement accepts $\rho_{i_j}$ with certainty. By the fidelity constraint, the probability that all the measurements for a given $j$ accept $\rho$ if $\rho \neq \rho_{i_j}$ is at most $(1-\epsilon)^{l} \le e^{-\epsilon l} \le \delta / (2k)$. By a union bound, the probability that any group of measurements incorrectly all accepts is at most $\delta / 2$. Thus the probability that the whole protocol fails is at most $\delta$. The overall number of copies of $\rho$ used is at most $k(l+1) = O((\|G\|/ \epsilon) (\log 1/\delta) \log(\|G\|/\delta))$ as claimed.
\end{proof}

Note the following additional points about Theorem \ref{thm:pure}:

\begin{enumerate}
\item For pure states, Theorem \ref{thm:mixed} is a corollary of Theorem \ref{thm:pure}. Letting $G$ be the Gram matrix of the states $\{\ket{\psi_i}^{\otimes k}\}$,
\[ \|G\| \le 1 + (n-1) \max_{i \neq j} |G_{ij}| = 1 + (n-1) \max_{i \neq j} |\braket{\psi_i|\psi_j}|^k. \]
If $|\braket{\psi_i|\psi_j}| \le 1-\epsilon$ for all $i \neq j$, this quantity becomes arbitrarily close to 1 for sufficiently large $k = O((\log n)/\epsilon)$.

\item It is not possible to obtain a similar result to Theorem \ref{thm:pure} for arbitrary mixed states (i.e.\ an upper bound only in terms of $\|G\|$), as can be seen by considering the states $\rho_i = I/n$. In this case one can calculate that $G = J \otimes (I/n)$, where $J_{ij} = 1$, $i,j\in \{1,\dots,n\}$; so $\|G\| = 1$, but the maximal worst-case success probability that can be achieved is $1/n$.

\item A case where Theorem \ref{thm:pure} is quite weak is a set of states whose pairwise inner products are all equal to some constant $c \in (0,1)$. Then $\|G\| = 1 + c(n-1)$, so Theorem \ref{thm:pure} states that the state discrimination problem could be solved with $\widetilde{O}(n)$ copies. It is obvious that this could be improved to $\widetilde{O}(\log n)$ copies, as taking $k$ copies maps $c \mapsto c^k$. But in this case we can explicitly calculate that $(\sqrt{G}_{ii})^2 \ge 1-c-O(1/n)$, so only $O(1)$ copies are required to achieve a high probability of success~\cite{montanaro07a}.
\end{enumerate}


\subsection*{Acknowledgements}

I would like to thank Aram Harrow, Andreas Winter and Jon Tyson for comments on previous versions. I acknowledge support from the QuantERA ERA-NET Cofund in Quantum Technologies implemented within the European Union's Horizon 2020 Programme (QuantAlgo project), EPSRC Early Career Fellowship EP/L021005/1, and EPSRC grant EP/R043957/1. This project has received funding from the European Research Council (ERC) under the European Union's Horizon 2020 research and innovation programme (grant agreement No.\ 817581). No new data were created during this study.



\bibliographystyle{plain}
\bibliography{../../thesis}

\end{document}